\documentclass[a4paper,UKenglish,cleveref, autoref, thm-restate]{lipics-v2019}
\pdfoutput=1
\usepackage{hyperref}
\hypersetup{
    colorlinks=true,
    linkcolor=black,
    citecolor=black,
    filecolor=black,
    urlcolor=black,
}

\title{A note on VNP-completeness and border complexity} 

\titlerunning{A note on VNP-completeness and border complexity} 

\author{Christian Ikenmeyer}{University of Liverpool}{christian.ikenmeyer@liverpool.ac.uk}{}{CI supported by DFG grant IK 116/2-1.}
\author{Abhiroop Sanyal}{Chennai Mathematical Institute}{abhiroop.sanyal@gmail.com}{}{}

\authorrunning{Christian Ikenmeyer and Abhiroop Sanyal} 

\Copyright{\mbox{~}} 

\ccsdesc[500]{Theory of computation~Algebraic complexity theory}
\keywords{algebraic complexity theory, VNP, border complexity, reductions, completeness, topology} 

\acknowledgements{We thank Michael Forbes for helpful insights about the relative power of reduction notions. Moreover, anonymous referees made us generalize our two main theorems.
We thank Josh Grochow for discussions that led to significant adjustments in Section~\ref{sec:topology}.}

\nolinenumbers 

\hideLIPIcs  

\newcommand{\pp}{\textnormal{\textup{\textsf{p}}}}
\newcommand{\cc}{\textnormal{\textup{\textsf{c}}}}
\newcommand{\HC}{\mathrm{HC}}
\newcommand{\VP}{\mathrm{VP}}
\newcommand{\VNP}{\mathrm{VNP}}
\newcommand{\VNPC}{\mathrm{VNPC}}
\newcommand{\VF}{\mathrm{VF}}
\newcommand{\VFC}{\mathrm{VFC}}
\newcommand{\VBP}{\mathrm{VBP}}
\newcommand{\VBPC}{\mathrm{VBPC}}
\newcommand{\VPC}{\mathrm{VPC}}
\newcommand{\VWaring}{\mathrm{VWaring}}
\newcommand{\dc}{\textup{\textrm{dc}}}

\usepackage{tikz}
\usetikzlibrary{arrows.meta}
\usepackage{amsthm}
\usepackage{mathtools}
\usepackage{color,soul}

\theoremstyle{definition}
\newtheorem*{example*}{Example ($\ast$)}

\usepackage{tikz}
\usetikzlibrary{arrows.meta}

\newcommand\blfootnote[1]{%
  \begingroup
  \renewcommand\thefootnote{}\footnote{#1}%
  \addtocounter{footnote}{-1}%
  \endgroup
}

\begin{document}
\maketitle

\blfootnote{2021-Dec-01}

\begin{abstract}
In 1979 Valiant introduced the complexity class VNP of p-definable families of polynomials, he defined the reduction notion known as p-projection and he proved that the permanent polynomial and the Hamiltonian cycle polynomial are VNP-complete under p-projections.

In 2001 Mulmuley and Sohoni (and independently B\"urgisser) introduced the notion of border complexity to the study of the algebraic complexity of polynomials.
In this algebraic machine model,
instead of insisting on exact computation,
approximations are allowed.
This gives VNP the structure of a topological space.
In this short note we study the set VNPC of VNP-complete polynomials. We show that the complement VNP $\setminus$ VNPC lies dense in VNP. Quite surprisingly, we also prove that VNPC lies dense in VNP. We prove analogous statements for the complexity classes VF, VBP, and~VP.

The density of VNP $\setminus$ VNPC holds for several different reduction notions:
p-projections, border p-projections, c-reductions, and border c-reductions.
We compare the relationships of the completeness notions under these reductions and separate most of the corresponding sets.
Border reduction notions were introduced by Bringmann, Ikenmeyer, and Zuiddam (JACM 2018).
Our paper is the first structured study of border reduction notions.
\end{abstract}

\maketitle

\section{Introduction}
Valiant's famous determinant versus permanent conjecture \cite{Val:79} states that the algebraic complexity class VBP (polynomials that can be written as determinants of polynomially large matrices of linear polynomials) is strictly contained in the class VNP (polynomials that can be written as Hamilton cycle polynomials\footnote{in characteristic $\neq 2$ the Hamilton cycle polynomial can be replaced by the permanent polynomial for all our statements.} of polynomially large matrices of linear polynomials, see Section~\ref{sec:actprelims}).
In 2001 Mulmuley and Sohoni \cite{gct1} in their Geometric Complexity Theory approach towards resolving Valiant's conjecture stated a strengthening of the conjecture ($\VNP\not\subseteq\overline{\text{VBP}}$) that is based on \emph{border complexity}, which was stated independently for circuits by B\"urgisser \cite[hypothesis (12)]{bue:04} ($\VNP\not\subseteq\overline{\text{VP}}$).
Border complexity was first studied in the area of fast matrix multiplication algorithms, where it led to significant algorithmic speedups \cite{bini1980approximate, Bini:80,schonhage1981partial,romani1982some,coppersmith1982asymptotic,strassen1986asymptotic,coppersmith1987matrix}.
The advantage of working with the closures of complexity classes is that this makes a large set of tools from algebraic geometry and representation theory available, see e.g.\ \cite{IB:18}.
The hope is that VBP and VNP can still be separated in this coarser setting.
Indeed, it is a major open question in geometric complexity theory whether or not $\text{VBP}=\overline{\text{VBP}}$, see \cite{grochow_et_al:LIPIcs:2016:6314}. If $\text{VBP}=\overline{\text{VBP}}$, then Valiant's conjecture must in principle be provable by algebraic geometry, provided it is true.
If $\VNP \subseteq \overline{\text{VBP}}$, then the Geometric Complexity Theory approach fails unsalvageably, while Valiant's conjecture could still be true.

\begin{figure}
\centering
\begin{tikzpicture}
\draw (0,2) ellipse (3cm and 3.5cm);
\draw (0,3.5) ellipse (1.5cm and 0.5cm);
\draw (0,0.5) ellipse (2cm and 1.5cm);
\draw (0,0.25) ellipse (1.5cm and 1cm);
\draw (0,0) ellipse (1cm and 0.5cm);
\node at (0,0.25) {VF};
\node at (0,1) {VBP};
\node at (0,1.75) {VP};
\node at (0,3.75) {VNPC};
\node at (0,5) {VNP};
\end{tikzpicture}
\caption{The known inclusions of the classical algebraic complexity classes.
VF is the class of families of polynomials with polynomially sized formulas,
VBP is the class of families of polynomials that can be written as polynomially large determinants of matrices of linear polynomials,
VP is the class of families of polynomials with polynomially sized circuits.
VNPC is the set of VNP-complete families. From a topological perspective such a depiction can be misleading, because VNPC lies dense in VNP and also VNP$\setminus$VNPC lies dense in VNP (under \pp-projections), see Theorem~\ref{thm:density}.}
\label{fig:compl}
\end{figure}
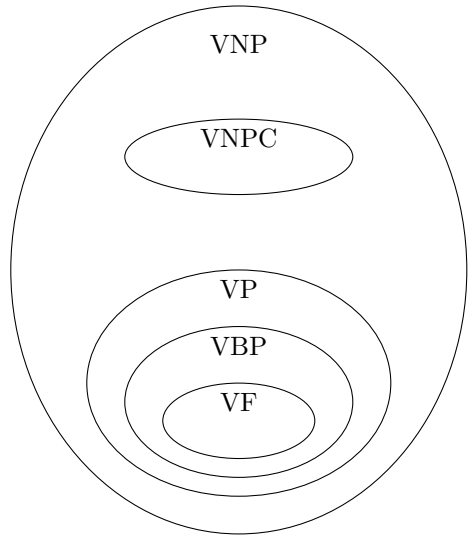
In Boolean complexity theory the relationship between complexity classes is often depicted in diagrams. An analogue for the classical algebraic complexity classes is given in Figure~\ref{fig:compl}.
In this paper we see that such a depiction presents misleading topological information:
We study the set of VNP-complete polynomials and its complement and see that
surprisingly both lie dense in VNP, see Theorem~\ref{thm:density}. We prove an analogous result for VF, VBP, and VP. This highlights that the topology is very coarse.

We take the methods for proving Theorem~\ref{thm:density} as a basis for studying VNP-completeness under different reduction notions. In particular we study border-$\pp$-projections, which were recently introduced in \cite{DBLP:journals/jacm/BringmannIZ18} with a focus on the border-$\pp$-projections of the iterated $2 \times 2$ matrix multiplication polynomial.
We get several separations of the power of different reduction notions in Theorem~\ref{thm:reductions}.
Our paper gives the first analysis of border reduction notions and their relative complexity in comparison to non-border reduction notions.
Moreover, our paper is the first to study border complexity analogues of the so-called oracle reductions (or \cc-reductions) that were introduced in \cite{burgisser1999structure}.

\section{Preliminaries}
\subsection{Algebraic Complexity Theory}\label{sec:actprelims}

Fix a field~$\mathbb{F}$.
An algebraic circuit is defined as a rooted\footnote{The digraph has a single node with out-degree 0.} directed acyclic graph which has its leaf nodes\footnote{Leaf nodes are the nodes with in-degree 0.} labelled with variables $\{x_{1}, x_{2}, \ldots , x_{n}\}$ or field constants, and the internal nodes labelled with $\times$ (``multiplication gates'') and $+$ (``addition gates'').
By induction over the circuit structure, each internal node computes a polynomial $f\in \mathbb{F}[x_{1}, x_{2}, \ldots, x_{n}]$.
The output of a circuit is defined as the polynomial computed at its root.
The size of an algebraic circuit is defined as the number of nodes in the circuit.

A sequence of natural numbers $(t) = (t_n)_{n \in \mathbb N}$ is called \emph{polynomially bounded} if there exists a polynomial function $p$ such that for all $n\in \mathbb N$ we have $t_n \leq p(n)$.
A sequence $(f) = (f_{n})_{n \in \mathbb N}$ of multivariate polynomials is defined to be a $\pp$-family if the number of variables in $f_{n}$ and the degree of $f_{n}$ are both polynomially bounded.
The complexity class $\VP$ is defined as the set of all $\pp$-families that have algebraic circuits whose size is polynomially bounded.
If we only allow \emph{skew} circuits, i.e., 
circuits for which each multiplication gate is adjacent to a leaf node, then we get the complexity class $\VBP$.
If instead we insist on the circuits to be rooted \emph{trees}, then we get the complexity class $\VF$.
We have $\VF \subseteq \VBP \subseteq \VP$. See \cite{toda1992classes} or \cite{Sapth} for an exposition about these classes, their different definitions, and their containment.

For fixed natural numbers $N$ and $M$,
a polynomial $f \in \mathbb{F}[x_{1}, x_{2}, \ldots , x_{N}]$ is said to be a \emph{projection} of another polynomial $g \in \mathbb{F}[y_{1}, y_{2}, \ldots, y_{M}]$ if $f = g(\alpha_{1}, \alpha_{2}, \ldots , \alpha_{M})$, where $\alpha_{i} \in \left\{x_{1}, x_{2}, \ldots , x_{N}\right\} \cup \mathbb{F}$. This is denoted by $f \leq g$.
A $\pp$-family $(f)$ is said to be the $\textsf{p}$-projection of another $\pp$-family $(g)$, denoted by $(f) \leq_{\textsf{p}} (g)$, if there is a polynomially bounded function~$t: \mathbb{N} \to \mathbb{N}$ such that $f_{n} \leq g_{t(n)}$ for all $n$\footnote{Note the non-uniformity of this notion of reduction as we get to choose the constants $alpha_i$ in the projection fresh (without complexity constraints) for every $n$.}.

$\pp$-projections were introduced by Valiant. Another natural example of reductions are $\cc$-reductions, an algebraic analogue of oracle complexity, which were introduced by B\"urgisser \cite{burgisser1999structure}.
The oracle complexity $L^{g}(f)$ of a polynomial $f$ with oracle $g$ is defined as the minimum size of a circuit with $+,  \times$, and $g$-oracle gates (the output at these gates is the computation of the polynomial $g$ on the input values. The arity of a $g$-oracle gate equals the number of variables in $g$.) that computes the polynomial $f$.
Consider two $\pp$-families $(f)$ and $(g)$. The $\pp$-family $(f)$ is said to be a $\cc$-reduction of $(g)$, denoted by $(f) \leq_{\cc} (g)$, if there exists a polynomially bounded function $t : \mathbb{N} \to \mathbb{N}$ such that $L^{g_{t(n)}}(f_{n})$ is polynomially bounded.
It is clear that $(f) \leq_\pp (g)$ implies $(f) \leq_\cc (g)$.

Let $\mathfrak S_n$ denote the symmetric group on $n$ symbols.
The determinant
is the sequence of homogeneous degree $n$ polynomials in $n^2$ many variables defined via
$\det_n \coloneqq \sum_{\pi\in\mathfrak S_n}\text{sign}(\pi)\prod_{i=1}^n x_{i,\pi(i)}$. Valiant showed that $(f) \in \VBP$ if and only if $(f) \leq_\pp (\det)$.

Let $C_{n} \subseteq \mathfrak S_n$ be the subset of length $n$ cycles. The Hamiltonian cycle family $(\HC)$ is the sequence of homogeneous degree $n$ polynomials $\HC_n$ on $n^2$ many variables defined via
$    \HC_n \coloneqq \sum_{\pi \in C_{n}} \prod_{i=1}^{n}x_{i,\pi(i)}.$
We define the class $\VNP$ as the set of all $\pp$-families $(f)$ that satisfy $(f)\leq_\pp (\HC)$. This is known to be equivalent to $(f)\leq_\cc (\HC)$.
Often a different definition is given that is closer in spirit to the counting complexity class $\#$P, where $\VNP$ is defined as a summation of a $\VP$ function over the Boolean hypercube, but Valiant \cite{Val:79} showed that these definitions are equivalent.
In particular, it is easy to see that $\VP \subseteq \VNP$.
Valiant's famous conjectures are $\text{VF} \neq \VNP$, $\VBP \neq \VNP$, and $\text{VP} \neq \VNP$. To find candidates of $\pp$-families in $\VNP$ that are outside of $\VF$, $\VBP$ or $\VP$, the following notion of $\VNP$-completeness is useful.

A $\pp$-family $(f)$ is defined to be $\VNP$-$\pp$-complete if $(f)\in\VNP$ and $(\HC) \leq_\pp (f)$.
The set of all $\VNP$-$\pp$-complete $\pp$-families is denoted by $\VNPC(\leq_\pp)$.
Analogously, a $\pp$-family $(f)$ is defined to be $\VNP$-$\cc$-complete if $(f)\in\VNP$ and $(\HC) \leq_\cc (f)$.
The set of all $\VNP$-$\cc$-complete $\pp$-families is denoted by $\VNPC(\leq_\cc)$.
It is clear that $\VNPC(\leq_\pp) \subseteq \VNPC(\leq_\cc)$. The fact that these two sets are different was established in \cite{IM:18} with a short argument.

The main motivation behind VNP-completeness comes from the following simple observation:
If we would find $(g)$ such that both $(g) \in \VNPC(\leq_\pp)$ and $(g) \in \text{VP}$, then for all $(f) \in \VNP$ we would have $(f) \leq_\pp (\HC) \leq_\pp (g)$, and by transitivity
$(f) \leq_\pp (g)$, which implies $(f) \in \text{VP}$, thus $\text{VP}=\VNP$. This observation similarly holds for $\VF$ and $\VBP$ instead of $\VP$.

As the natural variation of the above, a $\pp$-family $(f)$ is defined to be $\VF$-$\pp$-complete if $(f)\in\VF$ and $(g) \leq_\pp (f)$ for every $(g)\in\VF$.
The set of all $\VF$-$\pp$-complete $\pp$-families is denoted by $\VFC(\leq_\pp)$. The iterated $3\times 3$ matrix multiplication polynomial family is an example of an element in $\VFC(\leq_\pp)$ \cite{ben1992computing}.
A $\pp$-family $(f)$ is defined to be $\VBP$-$\pp$-complete if $(f)\in\VBP$ and $(g) \leq_\pp (f)$ for every $(g)\in\VBP$.
The set of all $\VBP$-$\pp$-complete $\pp$-families is denoted by $\VBPC(\leq_\pp)$. The determinant polynomial family is an example of an element in $\VBPC(\leq_\pp)$, see \cite{toda1992classes}.
A $\pp$-family $(f)$ is defined to be $\VP$-$\pp$-complete if $(f)\in\VP$ and $(g) \leq_\pp (f)$ for every $(g)\in\VP$.
The set of all $\VP$-$\pp$-complete $\pp$-families is denoted by $\VPC(\leq_\pp)$. There exist specific graph homomorphism polynomial families that are in \mbox{$\VPC(\leq_\pp)$}, see the recent \cite{mahajan2018some}.
This resolved a long-standing open problem of finding natural problems in~$\VPC(\leq_\pp)$.

\subsection{Border complexity and topology}\label{sec:topology}
In this section we define so-called \emph{border complexity analogues} to the reduction notions in Section~\ref{sec:actprelims}.
This was first done explicitly in \cite{DBLP:journals/jacm/BringmannIZ18}.
In the rest of the paper we assume that our base field $\mathbb F$ is $\mathbb Q$, $\mathbb R$, or~$\mathbb C$\footnote{This also means that the Hamilton cycle polynomial can be replaced by the permanent polynomial for all our statements \cite{Val:79}. We chose $\HC$ to be consistent with \cite{IM:18}, which also studies other fields.
In Lemma~\ref{lem:roots} we require the field to be large enough. In Lemma~\ref{2} we require the field to be a subfield of $\mathbb C$, but it is possible to phrase our density results in a field-independent way that is purely algebraic, see~\cite{IS:21v2}. We focus only on the most natural case of defining closures, which is over an infinite field such that we have a metric on~$P$.
}.
The ring $P$ of polynomials over $\mathbb F$ in any number of variables is a countably infinite dimensional metric vector space with the coefficient-wise Euclidean metric $\rho(f,g) \coloneqq ||f-g||$, where $||.||$ is the Euclidean norm of the vector of coefficients.
We define $R$ to be the vector space of sequences $(f) \coloneqq (f_n)_{n \in \mathbb N}$ of elements from~$P$:
\begin{equation*}
    R \coloneqq P^{\mathbb N}.
\end{equation*}
Let
\[
D \coloneqq (\mathbb R_{>0})^{\mathbb N}
\]
be the set of sequences of positive real numbers.
We endow $R$ with the \emph{box topology}\footnote{The usual product topology used in an early preprint of this paper \cite{IS:21} is too coarse to correctly capture the definitions of $\overline{\VBP}$ etc from the literature, which was kindly pointed out by Josh Grochow.
On the other hand, the closure in \cite{DBLP:journals/jacm/BringmannIZ18} is only defined for subsets of $R$ that are defined via an upper bound on an algebraic complexity measure, which provides no way to take the closure of the set of VNP-complete $\pp$-families, or of their complement.
Since we are the first to define a closure operator that works for all subsets of $R$, in Section~\ref{sec:boxtopology} we prove that our definition coincides with the defintion from the literature in all cases.}: The basic open sets are the sequences of open balls of radii $(\delta) \in D$ around elements $(f)\in R$, formally
\[
B_{(f),(\delta)} \coloneqq \{(g) \in R \mid \forall n\in\mathbb N: \rho(f_n,g_n)<\delta_n \}.
\]
A subset of $R$ is called \emph{open} in $R$ if it is the (possibly infinite) union of basic open sets. A subset $A$ of $R$ is called \emph{closed} in $R$ if $R\setminus A$ is open in $R$.
The \emph{closure} $\overline A$ of $A$ in $R$ is the smallest closed subset of $R$ containing $A$ \footnote{The closure notation with the overline was first introduced in \cite{BLMW:11}. The earlier work \cite{bue:04} used an underline, which is natural when thinking in terms of polynomially bounded border complexity measures.
}.
Since $\VNP \subseteq R$, we have that $\VF$, $\VBP$, $\VP$, and $\VNP$ are topological spaces using the subspace topology
\footnote{
The box topology construction also works with other topologies on~$P$, so that in principle it can be used over finite fields as well.
}.

In so-called sequential topological spaces the closure $\overline A$ can be described as the set of all limit points of sequences in $A$. But $R$ is not a sequential space (and hence not first-countable). So we have to use the more general notion of convergence of \emph{nets} instead of sequences.
We now recall the standard definition of net convergence and give a natural example in our setting.
A \emph{downward-directed set} $\mathscr D$ is a nonempty set with a reflexive and transitive binary relation~$\preceq$ (i.e., a preorder) with the additional property that
for all
$\varepsilon'\in \mathscr D$ and $\varepsilon''\in \mathscr D$ there exists $\varepsilon''' \in \mathscr D$ with $\varepsilon''' \preceq \varepsilon'$ and $\varepsilon'''\preceq \varepsilon''$.
\footnote{One important example that we will use often is the set of sequences of positive real numbers $\mathscr D = D$ with elementwise comparison. Usually in the definition of nets the sets are 
\emph{upward-directed} instead of downward-directed, but downward-directed sets fit better with our intuition of $\varepsilon$ getting smaller and smaller.}
For a topological space $X$ a \emph{net} $((g))$ in $X$ is a function $\mathscr D \to X$, where $\mathscr D$ is a downward-directed set.
We write $(g_\varepsilon)$ for its evaluation at $\varepsilon \in \mathscr D$.
We say a net $((g))$ \emph{converges} to a point $(f)\in X$ if for every open set $U$ containing $(f)$ there exists $\varepsilon \in \mathscr D$ such that for all $\varepsilon'\preceq \varepsilon$ we have $(g_{\varepsilon'}) \in U$. Note that it is equivalent to require this only for basic open sets $U$.
If $((g))$ converges to $(f)$, then we call $(f)$ a \emph{limit point} of $((g))$.
The key property of nets is that 
for every topological space $X$ and every subset $A \subseteq X$ we have that
$(f) \in \overline{A}$ if and only if there exists a net in $A$ that converges to $(f)$, see e.g.\ \cite[Ch.~2, 2 Theorem (b)]{kelley2017general}. \footnote{Note that if we change the definition of downward-directed sets to \emph{upward-directed} and let $\mathscr D$ be $\mathbb{N}$, then we recover the usual definition of sequences.}

\begin{example*}
For an illustrative example on how nets are used, fix the ground field to be $\mathbb C$ and recall that the \emph{Waring rank} of a homogeneous degree $n$ polynomial $f$ is defined as the smallest number $r$ such that homogeneous linear polynomials $\ell_i$ exist, $1\leq i \leq r$, with $f=\sum_{i=1}^r \ell_i^n$.
For example, the Waring rank of $x_1^2 x_2$ is at most 3, because $6 x_1^2 x_2 = (x_1+x_2)^3-(x_1-x_2)^3-(x_2)^3$.
Let $\VWaring_2$ denote the set of all sequences of $\pp$-families of homogeneous polynomials of Waring rank at most 2. Let $(f) \in R$ be the $\pp$-family of homogeneous degree $n$ polynomials defined via $f_n \coloneqq x_1^{n-1} x_2$.
Recall the downward-directed set $D \coloneqq (\mathbb R_{>0})^{\mathbb N}$.
For every $n\in \mathbb N$, $(\varepsilon) \in D$, we define
$g_{n,(\varepsilon)} \coloneqq \frac{1}{n\varepsilon_{n}}(x_1+ \varepsilon_n x_2)^n-\frac{1}{n\varepsilon_{n}}x_1^n$.
Clearly, for every fixed $(\varepsilon) \in D$ we have $(g_{(\varepsilon) }) \coloneqq (g_{n,(\varepsilon) })_{n \in \mathbb N} \in \VWaring_2$, so $((g))$ is a net in $\VWaring_2$.
To show that $(f) \in \overline{\VWaring_2}$ it remains to show that the net $((g))$ converges to $(f)$. This can be seen as follows.
For every $(\delta)\in D$ we can choose sufficiently small $(\varepsilon) \in D$ such that $(g_{(\varepsilon)}) \in B_{(f),(\delta)}$. Since for all $\varepsilon' \preceq \varepsilon$ and all $n$ we have $\rho(g_{n,(\varepsilon')},f_n) \leq \rho(g_{n,(\varepsilon)},f_n)$, it follows that 
for all $\varepsilon' \preceq \varepsilon$ we have $(g_{(\varepsilon')}) \in B_{(f),(\delta)}$.
Hence $((g))$ converges to $(f)$.
In the rest of this paper we will always use this simple form of the convergence of nets.
\hfill$\bullet$
\end{example*}

A subset $A\subseteq X$ is defined to be \emph{dense in $X$} if $\overline{A}=X$, i.e., every point in $X$ is a limit of a converging net of elements in $A$.
For all the sets $\VF$, $\VBP$, $\VP$, $\VNP$, it is a major open question if they are equal to their closure, see \cite{grochow_et_al:LIPIcs:2016:6314}.
In this paper we restrict our attention to the subspace topology on $\VNP \subseteq R$.

The notion of border projections is closely related to the concept of closures, but it does not require the use of nets.
A polynomial $f$ is called a \emph{border projection} of a polynomial $g$, if there exists a sequence $(g_1, g_2, \ldots)$ with $g_i \leq g$ for all $i$ and $\lim_{m \to \infty}g_m = f$. Then we write $f \trianglelefteq g$. Clearly $f \leq g$ implies $f \trianglelefteq g$.
We extend this definition to sequences of polynomials as follows. A sequence of polynomials $(f)$ is defined to be a \emph{border \pp-projection} of another sequence of polynomials $(g)$, denoted by $(f) \trianglelefteq_{\pp} (g)$, if there is a polynomially bounded function $t : \mathbb{N} \to \mathbb{N}$ such that $f_n \trianglelefteq g_{t(n)}$ for all $n$. \footnote{We remark that in the literature the class $\overline{\VBP}$ is usually defined via: $(f) \in \overline{\VBP}$ if and only if $(f) \trianglelefteq_\pp (\det)$. We prove in Section~\ref{sec:boxtopology} that this coincides with our definition of $\overline{\VBP}$. Analogously for $\VF$, $\VP$, $\VNP$ and their respective complete polynomials.
}

The border oracle complexity $\overline{L^{g}}(f)$ of a polynomial $f$ with oracle access to $g$ is defined as the smallest $r$ such that a sequence of polynomials $(g_1,g_2,\ldots)$ exists with $L^g(g_m)\leq r$ for all $m$ and $\lim_{m\to\infty}g_m = f$. Clearly $\overline{L^{g}}(f)\leq L^g(f)$ for all $f$ and $g$.
Consider two $\pp$-families of polynomials, $(f)$ and $(g)$. The $\pp$-family $(f)$ is said to be a border $\cc$-reduction of $(g)$, denoted by $(f) \trianglelefteq_{\cc} (g)$, if there exists a polynomially bounded function $t : \mathbb{N} \to \mathbb{N}$ such that $\overline{L^{g_{t(n)}}}(f_{n})$ is polynomially bounded.

A $\pp$-family $(f)$ is defined to be $\VNP$-border-$\pp$-complete if $(f)\in\VNP$ and $(\HC) \trianglelefteq_\pp (f)$.
The set of all $\VNP$-border-$\pp$-complete $\pp$-families is denoted by $\VNPC(\trianglelefteq_\pp)$.
Analogously, a $\pp$-family $(f)$ is defined to be $\VNP$-border-$\cc$-complete if $(f)\in\VNP$ and $(\HC) \trianglelefteq_\cc (f)$.
The set of all $\VNP$-border-$\cc$-complete $\pp$-families is denoted by \mbox{$\VNPC(\trianglelefteq_\cc)$}.

In geometric complexity theory the commonly used type of reduction is similar to $\pp$-projections and border-$\pp$-projections, but instead of replacing variables with constants and variables, variables are replaced by affine linear polynomials \cite{BLMW:11}. All proofs in this paper also work with this version of $\pp$-projections and border-$\pp$-projections.

\section{Main results}
Using a fairly elementary proof we obtain the following surprising density results.
\begin{theorem}\label{thm:density}
Let $\leq$ be any of $\leq_\pp$, $\leq_\cc$, $\trianglelefteq_\pp$, or $\trianglelefteq_\cc$.
Let $\lesssim$ be any of $\leq_\pp$, $\trianglelefteq_\pp$. Then:
\begin{itemize}
\item the set $\VNPC(\leq)$ is a dense subset of $\VNP$ and the complement $\VNP \setminus \VNPC(\lesssim)$ is a dense subset of~$\VNP$.
\end{itemize}
Analogously,
\begin{itemize}
\item $\VFC(\leq)$ and $\VF \setminus \VFC(\lesssim)$ are both dense in $\VF$,
\item $\VBPC(\leq)$ and $\VBP \setminus \VBPC(\lesssim)$ are both dense in $\VBP$,
\item $\VPC(\leq)$ and $\VP \setminus \VPC(\lesssim)$ are both dense in $\VP$.
\end{itemize}
\end{theorem}
The multitude of density statements in Theorem~\ref{thm:density} suggests that the topology on $R$ is suprisingly coarse (i.e., has few open sets).
We leave it as an open problem if $\VNP \setminus \VNPC(\leq)$ is dense in $\VNP$ for the other two reduction notions $\leq_\cc$ and~$\trianglelefteq_\cc$.
Oracle complexity is too coarse to study $\VF$-completeness, $\VBP$-completeness, or $\VP$-completeness, because a polynomially sized circuit does not gain any benefit from having oracle access to a polynomially sized formula or a polynomially sized determinant, or another polynomially sized circuit.

As a second result we initiate the comparison of classical and border reduction notions.
We give an almost complete separation of the sets of $\VNP$-complete polynomials under the different reduction notions as follows.
\begin{theorem}\label{thm:reductions}We have the following diagram of inclusions (solid arrows) and non-inclusions (dashed arrows).
\begin{center}
\begin{tikzpicture}
\node (pp) at (0,0) {$\VNPC(\leq_\pp)$};
\node (cc) at (3,-1) {$\VNPC(\leq_\cc)$};
\node (tripp) at (3,1) {$\VNPC(\trianglelefteq_\pp)$};
\node (tricc) at (6,0) {$\VNPC(\trianglelefteq_\cc)$};
\draw[-Triangle] (pp.east)+(0,-0.1) to [bend left=10] (cc);
\draw[dashed, -Triangle] (cc) to [bend left=10] node[midway,below] {\tiny\cite{IM:18}} (pp);
\draw[-Triangle] (cc) to [bend right=10] (tricc);
\draw[dashed, -Triangle] (tricc) to [bend right=10] node[midway,right]{\tiny Cor.~\ref{cor:HCsquareccomplete} and Lem.~\ref{2}} (tripp);
\draw[-Triangle] (tripp) to [bend right=10] (tricc);
\draw[-Triangle] (pp) to [bend right=10] (tripp);
\draw[dashed, -Triangle] (tripp) to [bend right=10] node[midway,above left] {\tiny Lem.~\ref{lem:IM18} and Lem.~\ref{lem3}} (pp);
\draw[dashed, -Triangle] (cc) to node[midway, fill=white] {\tiny Cor.~\ref{cor:HCsquareccomplete} and Lem.~\ref{2}} (tripp);
\end{tikzpicture}
\end{center}
Analogously for $\VF$, $\VBP$, and $\VP$ instead of $\VNP$.
\end{theorem}
The inclusions (solid arrows) are obvious consequences of the relations among $\leq_\pp$, $\leq_\cc$ etc. The non-inclusions are proved in the respective lemmas as annotated in the figure. We only treat the case $\VNP$, as the other classes are handled in exactly the same way, with replacing $(\HC)$ by the respective complete $\pp$-family. One caveat: Lemma~\ref{lem:roots} does not work out of the box for $\VF$. We treat this case in Section~\ref{sec:formulas}.

\section{Related work}
The relative power of reductions in the Boolean world has been analyzed for example in \cite{OH:02, HP:07}. The notion of $\cc$-reductions in algebraic complexity theory is relatively new \cite{burgisser1999structure}. The difference between $\pp$-projections and $\cc$-reductions plays a prominent role in \cite{mahajan2018some}.
The relative power of algebraic reduction notions has been studied before:
\cite{IM:18} show with a short argument that \[\VNPC(\leq_\pp) \subsetneqq \VNPC(\leq_\cc).\]
They do not study border complexity though.

Border complexity has already been an object of study in algebraic complexity theory for bilinear maps since 1980 (see \cite{Bini:80}) and is still a very active area of research today \cite{lo:15, lm:18}.
The study of border complexity for polynomials has recently gained significant momentum, see for example \cite{grochow_et_al:LIPIcs:2016:6314, Kum:18, DBLP:journals/jacm/BringmannIZ18,  blser_et_al:LIPIcs:2020:12573}.
In fact, \cite{DBLP:journals/jacm/BringmannIZ18} prove that $\text{VFC}(\leq_\pp) \subsetneqq \text{VFC}(\trianglelefteq_\pp)$.
This result makes use of a result on width-2 ABPs that is obtained from a fairly involved analysis \cite{aw:16}.

\section{Proof of Theorem~\ref{thm:density}}
We start with an observation of \cite{IM:18} that we state simultaneously for $\VNP$, $\VF$, $\VBP$, and~$\VP$. It is clear that the result holds in much higher generality.

\begin{proposition}\label{prop:IM18}
Suppose $(f) \in \VNP$ \textup{[}$\VF$, $\VBP$, $\VP$\textup{]} and each $f_n$ has the following form:
$
f_{n} = q (rg_{n} + c_n g_{n}^{2})
$
for some $\pp$-family $(g)$, some fixed polynomial $r$ of any degree $d_2$, some fixed polynomial $q$ of even degree $d_1$ that is also a perfect square, and some sequence $(c)$ of nonzero constants.
Then $(f) \notin \VNPC(\leq_\pp)$ \textup{[}$\VFC(\leq_\pp)$, $\VBPC(\leq_\pp)$, $\VPC(\leq_\pp)$\textup{]}.
\end{proposition}
\begin{proof}
The proof is a minor generalization of \cite[Lemma~3.2]{IM:18}. Consider a univariate polynomial $s(y)$ of odd degree $M > d_{1} + 2d_{2}$.
We claim that $s$ cannot be written as a projection of $f_{n}$, for any $n$.
Let $\deg_{y}(h)$ denote the degree of a polynomial $h$ in the variable $y$, when considered as a polynomial over the polynomial ring in all its other constituent variables.
Let $\gamma$ be any linear projection map in the sense of $\leq_\pp$. Then, $\deg_{y}(\gamma(f_{n})) \leq \textsf{max$[\deg_{y}(\gamma(q\cdot r\cdot g_{n})), \deg_{y}(\gamma(q\cdot g_{n}^{2}))]$}$. Also, note that $\deg_{y}(\gamma(q)) \leq d_{1}$ and $\deg_{y}(\gamma(r)) \leq d_{2}$.

If $\deg_{y}(\gamma(g_{n})) \leq d_{2}$, then $\deg_{y}(\gamma(f_{n})) \leq d_{1} +2d_{2} $. If $\gamma(f_{n}) = s(y)$, this contradicts the fact that $s(y)$ has degree~$M$. 

Otherwise, $\deg_{y}(\gamma(g_{n})) > d_{2}$ and hence
$\deg_{y}(\gamma(f_{n})) = \deg_{y}(\gamma(q)) + \max\{\deg_{y}(\gamma(rg_n)),\deg_{y}(\gamma(c_n g_n^2))\} = \deg_{y}(\gamma(q\cdot g_{n}^{2}))$, because $\deg_{y}(\gamma(r)) \leq d_{2}$. But $q\cdot g_{n}^{2}$ is a perfect square polynomial, hence $\deg_y(\gamma(f_n))$ is even, but $s(y)$ has odd degree. Hence, $\gamma(f_{n}) \neq s(y)$.

Thus, the constant sequence $(s)$ cannot be written as a $\textsf{p}$-projection of $f$. Hence, $(f) \notin \VNPC(\leq_\pp)$ [$\VFC(\leq_\pp)$, $\VBPC(\leq_\pp)$, $\VPC(\leq_\pp)$].
\end{proof}

In a simpler setting the above ideas work also for border reductions, as the following proposition shows.
\begin{proposition}[border reduction version of {Prop.~\ref{prop:IM18}}]\label{prop:bordernew}
Let $g$ be a non-constant polynomial with $g \trianglelefteq f + c f^2$ for a nonzero field constant~$c$. Then $\deg(g)$ is even.
\end{proposition}
\begin{proof}
A classical theorem by Lehmkuhl and Lickteig \cite{LL:89} over $\mathbb F \in \{\mathbb C, \mathbb R\}$ (published for order 3 tensors, but frequently its generalization to arbitrary polynomials is used, see e.g.\ \cite{bue:04,forbes2018pspace,guo2018algebraic,sinhababu2019power,huttenhain2017geometric,medini2021hitting}) implies that the approximation of $f$ can be written in terms of Laurent polynomials:
$f' \coloneqq (f + c f^2)(\overline{\alpha})$,
\[
    \lim_{\varepsilon \to 0}\big( f'(\alpha_{1}, \alpha_{2}, \ldots , \alpha_{N})\big) = g
\]
for some $\alpha_{i}$ $\in$ $\mathbb{F}[\varepsilon,\varepsilon^{-1}] \cup \overline{\mathbf{x}}$.
If all homogeneous components of $f(\overline{\alpha})$ converge for $\varepsilon\to 0$, then $\lim_{\varepsilon\to 0}(f') = \lim_{\varepsilon\to 0}f(\overline{\alpha}) + c (\lim_{\varepsilon\to 0}f(\overline{\alpha}))^2$.
So the top homogeneous part of $\lim_{\varepsilon\to 0}(f')$ is (up to rescaling with $c$) the square of the top homogeneous part of $\lim_{\varepsilon\to 0}(f(\overline{\alpha}))$. Hence in this case the degree is even.

We now treat the case where there exists a homogeneous component of $f(\overline{\alpha})$ that diverges for $\varepsilon\to 0$.
Let $d$ be the largest degree of such a component.
Assume for now that $d\neq 0$.
Then, the homogeneous degree $2d$ component of $f^2(\overline{\alpha})$ must diverge and this cannot be cancelled out by the other homogeneous components which are of lower degree. This is a contradiction to the convergence of the sum. Hence, $d = 0$.

Now assume that the homogeneous degree 0 component $f(\overline{\alpha})_0$ of $f(\overline{\alpha})$ diverges for $\varepsilon\to 0$ and let $M<0$ be the smallest exponent of $\varepsilon$ in the univariate Laurent polynomial $f(\overline{\alpha})_0$.
This implies that the homogeneous degree 0 component of $f'$ is a Laurent polynomial in $\varepsilon$ that has a nonzero term of degree $2M$. Hence $f^{\prime}$ diverges, which is a contradiction to our initial assumption that $\lim_{\varepsilon\to 0}f^{\prime} = g$.
\end{proof}

\begin{proof}[Proof of Theorem \ref{thm:density}]
First we prove that $\VNPC(\leq_\pp)$ is dense in $\VNP$.
Let $\left(f\right) \in \VNP$ be arbitrary.
Recall the downward directed set $D \coloneqq (\mathbb R_{>0})^{\mathbb N}$.
For every $n\in \mathbb N$, $(\varepsilon) \in D$, we define
\begin{equation*}
    g_{n,(\varepsilon)} \coloneqq f_{n} + \varepsilon_n y_n\left(\mathrm{HC}_{n} - f_{n}\right),
\end{equation*}
where each $y_n$ is a variable that is unused by $f_{n}$ and $\mathrm{HC}_{n}$ for all $n$.
For every fixed $(\varepsilon) \in D$ we have $(g_{(\varepsilon) }) \in \VNP$, because $\VNP$ is closed under taking finitely many sums and products.
Moreover, we have $(g_{(\varepsilon)}) \in \VNPC(\leq_\pp)$, which can be seen as follows.
Indeed, $(\HC) \leq_\pp (g_{(\varepsilon)})$, because under the projection map $\gamma: y_n \mapsto \frac{1}{\varepsilon_n}$ we have $\gamma(g_{n,(\varepsilon)})=\HC_n$ and thus $\gamma(\, (g_{(\varepsilon)})\,) = (\HC)$.
Hence $((g))$ is a net in $\VNPC(\leq_\pp)$.

The net $((g))$ converges to $(f)$, which can be seen in the same way as in Example~($\ast$):
For every $(\delta)\in D$ we can choose sufficiently small $(\varepsilon) \in D$ such that $(g_{(\varepsilon)}) \in B_{(f),(\delta)}$. Since for all $\varepsilon' \preceq \varepsilon$ and all $n$ we have $\rho(g_{n,(\varepsilon')},f_n) \leq \rho(g_{n,(\varepsilon)},f_n)$, it follows that 
for all $\varepsilon' \preceq \varepsilon$ we have $(g_{(\varepsilon')}) \in B_{(f),(\delta)}$.

The result for the other reduction types is immediate, because
$\textsf{p}$-projections are the weakest notion of reduction we consider, in particular
$\VNPC(\leq_\pp) \subseteq \VNPC(\leq_\cc)$,
$\VNPC\nolinebreak(\nolinebreak\leq_\pp\nolinebreak)\nolinebreak \subseteq \VNPC\nolinebreak(\nolinebreak\trianglelefteq_\pp\nolinebreak)$,
and $\VNPC(\leq_\pp) \subseteq \VNPC(\trianglelefteq_\cc)$.

For the other part, let $(f)\in\VNP$ be arbitrary.
For every $n\in \mathbb N$, $(\varepsilon)\in D$, we define 
\begin{equation*}
    h_{n,(\varepsilon)} \coloneqq f_{n} + \varepsilon_n f_{n}^{2}
\end{equation*}
Obviously $(h_{(\varepsilon)}) \in \VNP$, but
according to Proposition~\ref{prop:bordernew} we have
that the constant $\pp$-family of the degree 1 polynomial $(x)$
has $(x) \not\trianglelefteq_\pp \left(h_{(\varepsilon)}\right)$ and hence 
$\left(h_{(\varepsilon)}\right) \notin \mbox{$\VNPC(\trianglelefteq_\pp)$}$.
Hence $((h))$ is a net in $\VNP \setminus \VNPC(\trianglelefteq_\pp)$.
Clearly $((h))$ converges to $(f)$ with the same argument as for $((g))$ in the first part of this proof.
Thus,
$\VNP \setminus \VNPC(\trianglelefteq_\pp)$ is dense in $\VNP$.

The analogous statements about $\VF$, $\VBP$, and $\VP$
that are claimed in the theorem are proved in exactly the same way by replacing $\HC_n$ by a complete polynomial family for the respective class.
\end{proof}

\section{Proof of Theorem~\ref{thm:reductions}}
We start with a classical lemma about taking roots.
\begin{lemma}\label{lem:roots}
Suppose $g = f^r$ for some $f \in \mathbb{F}[\overline{\mathbf{x}}]$ of degree $d$ and constant $r$, where $\overline{\mathbf{x}}$ denotes a set of variables. Then $f$ can be computed by a $g$-oracle circuit of size $poly(d)$. 
In particular, if $g_n=f_n^r$ for all $n$, then
$(f) \leq_\cc (g)$.
\end{lemma}
\begin{proof}
(The proof follows that of  a special case of \cite{DBLP:conf/stoc/Kaltofen87}, the proof technique borrows from \cite{DBLP:conf/stoc/Valiant79} and \cite{DuttaSaxena}). Consider a polynomial $g = f^{r}$, where $f$ has degree $d$. Notice that for every infinite field $\mathbb{F}$ and every nonzero polynomial $f$ over $\mathbb{F}[\overline{\mathbf{x}}]$, there exists $\mathbb{\overline{\alpha}}$ $\in$ $\mathbb{F}^{\vert \overline{\mathbf{x}}\vert}$, such that $f(\overline{\alpha}) \neq 0$. Also, shifting the variables $f(\overline{\mathbf{x}}) \mapsto f(\overline{\mathbf{x}}+\alpha)$
is an invertible operation
since you may re-shift
at the input nodes of the circuit.  Thus, by appropriately shifting
we may assume w.l.o.g.\ that $g(\mathbf{0}) \neq 0$.
Rescaling at the output node is also an invertible operation, so we may assume w.l.o.g.\ that $g(\mathbf{0}) =1$.
We can write:
\begin{equation*}
    f = (1 + (g-1))^{1/r}.
\end{equation*}
Using the binomial theorem for rational exponents, this gives us:
$\left(1 + (g-1)\right)^{1/r} =$
\begin{equation*}
      1 + \frac{1}{r}(g-1) + {1/r \choose 2}(g-1)^{2} + \cdots + {1/r \choose d}(g-1)^{d} + \cdots
\end{equation*}
Since $g(\mathbf{0}) = 1$, then $g-1 = 0$ mod $(\overline{\mathbf x})$. Thus, $(g-1)^{i}$ has only monomials of degree larger than $d$ for $i \geq d+1$. So, $f = $
\begin{equation*}
     1 + \frac{1}{r}(g-1) + {1/r \choose 2}(g-1)^{2} + \cdots + {1/r \choose d}(g-1)^{d} \hspace{1.5 mm} \text{mod} \!\left(\!\overline{\mathbf{x}}^{d+1}\!\right)
\end{equation*}
where $\overline{\mathbf{x}}^{d+1}$ denotes the set of all monomials of degree $d+1$. We have the oracle circuit for $g$. The modular operation can be done via Strassen's homogenization trick \cite{Strassen1973}. Specifically, each homogeneous part of $f$ can be written as a linear combination of  $(d+1)$ \textsf{p}-projections of~$g$. Thus, computing roots using oracle gates is possible with circuits of size poly$(d)$. This proves the first part.
The second part follows from the fact that $\pp$-families have polynomially bounded degrees.
\end{proof}

As an immediate corollary we obtain:
\begin{corollary}\label{cor:HCsquareccomplete}
$(\HC^2) \in \VNPC(\leq_\cc)$.
\end{corollary}
\begin{proof}
By Lemma~\ref{lem:roots} we have $(\HC) \leq_\cc (\HC^2)$.
Since $(\HC) \in \VNPC(\leq_\cc)$, for every $(f) \in \VNP$ we have
$(f) \leq_\cc (\HC)$.
By transitivity we have $(f) \leq_\cc (\HC^2)$. Therefore
$(\HC^2) \in \VNPC(\leq_\cc\nolinebreak)$.
\end{proof}

\begin{lemma}\label{2}
$(\mathrm{HC}^{2}) \not\in \VNPC(\trianglelefteq_\pp)$.
\end{lemma} 
\begin{proof}
We claim that every border-$\textsf{p}$-projection of $\mathrm{HC}_{n}^{2}$ is the square of a polynomial. The lemma follows from this fact, because for example the polynomial $x^3$ is not a square, but the constant $\pp$-family $(x^3)$ is in $\VNP$.

Let $f \in \mathbb{F}[\overline{\mathbf{x}}]$ be a border-$\textsf{p}$-projection of $\mathrm{HC}_{n}^{2}$.
We use \cite{LL:89} again:
\begin{equation}\label{convergeeq}
    \lim_{\varepsilon \to 0} \mathrm{HC}_{n}^{2}(\alpha_{1}, \alpha_{2}, \ldots , \alpha_{n^2}) = f
\end{equation}
for some $\alpha_{i}$ $\in$ $\mathbb{F}[\varepsilon,\varepsilon^{-1}] \cup \overline{\mathbf{x}}$.
We interpret $\mathrm{HC}_{n}(\alpha_{1}, \alpha_{2}, \ldots , \alpha_{n^2})$ as a univariate Laurent polynomial in $\varepsilon$ with coefficients in $\mathbb F[\overline{\mathbf{x}}]$. Let $M$ be the smallest exponent of~$\varepsilon$.
If $M<0$, then the univariate Laurent polynomial $\mathrm{HC}_{n}^2(\alpha_{1}, \alpha_{2}, \ldots , \alpha_{n^2})$ has a nonzero degree $2M$ term, which means that $\mathrm{HC}_{n}^2(\alpha_{1}, \alpha_{2}, \ldots , \alpha_{n^2})$ does not converge for $\varepsilon~\to~0$. This is ruled out by \eqref{convergeeq}. Therefore $M \geq 0$ and hence $L \coloneqq \lim_{\varepsilon\to 0}\mathrm{HC}_{n}(\alpha_{1}, \alpha_{2}, \ldots , \alpha_{n^2})$ exists.
Therefore, by continuity we have $f = \lim_{\varepsilon\to 0}\mathrm{HC}_{n}^2(\alpha_{1}, \alpha_{2}, \ldots , \alpha_{n^2}) = (\lim_{\varepsilon\to 0}\mathrm{HC}_{n}(\alpha_{1}, \alpha_{2}, \ldots , \alpha_{n^2}))^2 = L^2$.
\end{proof}

We now construct $(P) \in \VNPC(\trianglelefteq_\pp) \setminus \VNPC(\leq_\pp)$ via
\begin{equation*}
    P_n \coloneqq z^2(y\mathrm{HC}_n + y^2 \mathrm{HC}_n^2)
\end{equation*}
where $y$ and $z$ are variables outside the set of variables in $\mathrm{HC}_{n}$, for all~$n$.
\begin{lemma}\label{lem:IM18}
$(P) \not\in \VNPC(\leq_\pp)$
\end{lemma}
\begin{proof}
This is a direct consequence of Proposition~\ref{prop:IM18} with $g_n=y \HC_n$, $c_n=r=1$, $q=z^2$.
\end{proof}
\begin{lemma}\label{lem3}
$(P)$ $\in$ $\VNPC(\trianglelefteq_{\pp})$.
\end{lemma}
\begin{proof}
Consider the projection map $\gamma_\varepsilon$ defined by:
\begin{equation*}
    y \mapsto \varepsilon^2 \hspace{5 mm} \text{and} \hspace{5 mm} z \mapsto \tfrac{1}{\varepsilon}
\end{equation*}

For every $n \in \mathbb N$ we have $\gamma_\varepsilon(P_n) \leq P_n$ and
$\gamma_\varepsilon(P_n) = \mathrm{HC}_{n} + \varepsilon^{2}\mathrm{HC}_{n}^{2}$.
Hence $\HC_n \trianglelefteq P_n$.
Therefore 
$(\HC) \trianglelefteq_\pp (P)$, which finishes the proof.
\end{proof}

\subsection{Arithmetic formulas}\label{sec:formulas}
Lemma~\ref{lem:roots} does not directly work for arithmetic formulas.
In order to construct a $\pp$-family in $\VFC(\leq_\cc)\setminus \VFC(\trianglelefteq_\pp)$ we proceed as follows.
Let $(f) \in \VFC(\leq_\cc)$, for example iterated $3 \times 3$ matrix multiplication.
Define $(Q)$ via $Q_n \coloneqq (x+yf_n)^2$.
We have $(Q) \notin \VNPC(\trianglelefteq_\pp)$ in the same way as Lemma~\ref{2}.
Moreover, setting $x=0$ and $y=1$ we obtain $\beta_1 \coloneqq f_n^2$.
If instead we set $x=1$ and $y=1$, then we obtain
$\beta_2 \coloneqq f_n^2 + 2f_n + 1$.
Hence $f_n = \frac 1 2 (\beta_2-\beta_1-1)$. Thus $(f)\leq_\cc(Q)$ and hence $(Q) \in \VFC(\leq_\cc)$.

\section{The box topology is the correct topology}
\label{sec:boxtopology}

In this section we prove that the closure of the classical classes $\VF$, $\VBP$, $\VP$, and $\VNP$ in $R$ coincides with the classical definitions of $\overline{\VF}$, $\overline{\VBP}$, $\overline{\VP}$, and $\overline{\VNP}$ from the literature.

Given a polynomial $f$ with $v$ many variables and degree~$d$. Let $V$ denote the finite dimensional vector space of polynomials of degree at most $d$ in $v$ many variables, endowed with the standard Euclidean metric~$\rho$.
Let $\dc(f)$ denote the \emph{determinantal complexity} of $f$, i.e, the number of rows of the smallest square matrix $M$ with affine linear entries with $\det(M)=f$.
Valiant proved that that $\dc(f)$ is always finite.
Let $\underline{\dc}(f)$ denote the smallest $r$ such that there exists a sequence $(f_m)$ of polynomials with $\dc(f_m) \leq r$ such that $\lim_{m\to\infty}f_m = f$. The set of $\pp$-families $(f)$ where $\underline{\dc}( \ (f) \ )$ is polynomially bounded is called $\overline{\VBP}$ in the literature. We call it $\VBP'$ here to make it possible to distinguish it from the closure of $\VBP$.
We define $\VF'$, $\VP'$, and $\VNP'$ analogously via their respective complete polynomials.

\begin{proposition}
$\VF' = \overline{\VF}$ and
$\VBP' = \overline{\VBP}$ and
$\VP' = \overline{\VP}$ and
$\VNP' = \overline{\VNP}$.
\end{proposition}
\begin{proof}
We only show the argument for $\VBP$. The other cases are analogous.

Let $(f)\in\VBP'$. This means that $\underline{\dc}( \ (f) \ )$ is polynomially bounded, say $\underline{\dc}( \ (f) \ ) = p(n)$ for a polynomially bounded function $p$.
Then for every $n$ and every $\varepsilon_n \in \mathbb R_{>0}$ there exists $g_{n,\varepsilon_n}$ with $\rho(f_n,g_{n,\varepsilon_n})<\varepsilon_n$ and $\dc(g_{n,\varepsilon_n})\leq p(n)$.
Taking the downward directed set $D\coloneqq (\mathbb R_{>0})^{\mathbb N}$ and setting $g_{n,(\varepsilon)} \coloneqq g_{n,\varepsilon_n}$,
this defines a net $((g))$ in $\VBP$. It remains to show that $((g))$ converges to $(f)$. Let $(\delta)\in D$. Then $(g_{(\delta)}) \in B_{(f),(\delta)}$ and for all $(\varepsilon) \preceq (\delta)$ we also have $(g_{(\varepsilon)}) \in B_{(f),(\delta)}$ by the definition of $g_{n,(\varepsilon)}$. Hence $((g))$ converges to $(f)$.

For the other direction we introduce a new concept.
For every $(f) \in R$ we define a \emph{threshold distance sequence} $(\delta^\textup{thresh}) \in D$ by defining the $n$-th entry $\delta^\textup{thresh}_n$ as follows.
Let $v$ be the number of variables of $f_n$ and let $d$ be the degree of $f_n$. Let $V$ denote the finite dimensional vector space of polynomials of degree at most $d$ in $v$ many variables. Let $V_r \coloneqq \{ g \in V \mid \underline{\dc}(g)\leq r \}$. Each set $V_r$ is closed in $V$, in fact $V_r = \overline{\{ g \in V \mid \dc(g)\leq r \}}$.
We have $f_n \in V_r$ for all $r \geq \underline{\dc}(f)$, and $f_n \notin V_r$ otherwise.
We define $\delta^\textup{thresh}_n$ as the minimum distance between $f_n$ and $V_r$ for $r < \underline{\dc}(f)$, formally
$\delta^\textup{thresh}_n \coloneqq \min\{\rho(f_n,V_r) \mid r < \underline{\dc}(f)\}
$ \footnote{For determinantal complexity this equals $\rho(f_n,V_{\underline{\dc}(f)-1})$, but that is not important here.}.
Here we use the usual definition of a distance of a point to a set: $\rho(f_n,V_r)\coloneqq\inf\{\rho(f_n,g)\mid g \in V_r\}$.
Note that $\delta^\textup{thresh}_n> 0$, because each $V_r$ is a closed set (a point has distance 0 to a closed set if and only if that set contains the point).

Let $(f)\in\overline{\VBP}$. Then there exists a downward directed set $(\mathscr D,\preceq)$ and a net $((g)): \mathscr D \to \VBP$ that converges to $(f)$.
Hence for every $(\delta)\in D$ there exists $\varepsilon\in \mathscr D$ such that for all $\varepsilon'\preceq \varepsilon$ we have $(g_{\varepsilon'}) \in B_{(f),(\delta)}$.
In particular this holds for $(\delta^\textup{thresh})$, so $(g_{\varepsilon}) \in B_{(f),(\delta^\textup{thresh})}$ and $\dc( \ (g_{\varepsilon}) \ )$ is polynomially bounded.
By the definition of $\delta^\text{thresh}$ it follows that for all $n$ we have $\underline{\dc}(f_n) \leq \dc(g_{n,\varepsilon})$.
Hence $\underline{\dc}( \ (f) \ )$ is polynomially bounded, i.e., $(f)\in \VBP'$.
\end{proof}

It is clear that this proof technique is not limited to $\VF$, $\VBP$, $\VP$, and $\VNP$, but it works for every algebraic complexity class that is defined by circuits of a specific kind growing at a specific rate.

\end{document}